\documentclass[letterpaper, 10 pt, conference]{ieeeconf}  

\IEEEoverridecommandlockouts                              

\usepackage{graphics} 
\usepackage{epsfig} 
\usepackage{times} 
\usepackage{amsmath} 
\usepackage{amssymb}  
\usepackage{subfig}
\usepackage{graphicx}
\usepackage{epstopdf}
\usepackage{xcolor}
\usepackage{fancyhdr}

\newtheorem{assumption}{Assumption}
\newtheorem{lemma}{Lemma}
\newtheorem{definition}{Definition}
\newtheorem{remark}{Remark}
\newtheorem{proposition}{Proposition}
\newtheorem{theorem}{Theorem}
\newtheorem{rules}{Rule}

\title{\LARGE \bf
Disturbance Rejection Control under Nested Signal Temporal Logic Specifications: A Recursive Design Approach
}

\author{Yuzhang Peng, Jiaqi Yan, Kurui Ze, Wei Wang, and Yu Pan%
\thanks{*This work was supported by National Natural Science Foundation of China under Grants 62573023, 625B2017 and 62373019,
the Open Research Project of the State Key Laboratory of Industrial Control Technology, China, under Grant ICT2026B14,
and in part by the Fundamental Research Funds for the Central Universities under Grants 501RCQD2025103003 and 501XSKC2025103001.}
\thanks{Yuzhang Peng and Kurui Ze are with the School of Automation Science and Electrical Engineering, Beihang University, Beijing 100191, China.
        {\tt\small pengyz@buaa.edu.cn}; {\tt\small kr\_ze@buaa.edu.cn}}
\thanks{Jiaqi Yan and Wei Wang are with the School of Automation Science and Electrical Engineering, Beihang University, Beijing 100191, China and the Hangzhou Innovation Institute, Beihang University, Hangzhou 310051, China.
        {\tt\small jqyan@buaa.edu.cn}; {\tt\small w.wang@buaa.edu.cn}}
\thanks{Yu Pan is with the State Key Laboratory of Industrial Control Technology, Zhejiang University, Hangzhou 310027, China.
        {\tt\small ypan@zju.edu.cn}}
}

\begin{document}
\bstctlcite{IEEEexample:BSTcontrol}

\maketitle

\thispagestyle{fancy}
\pagestyle{empty}

\fancyhf{}
\fancyhead[C]{\small\textit{Accepted for presentation at the 65th IEEE Conference on Decision and Control (CDC), Honolulu, Hawaii, USA, December 15--18, 2026}}
\renewcommand{\headrulewidth}{0pt}

\begin{abstract}
For the control synthesis problem under signal temporal logic (STL) specifications, control barrier functions (CBFs) serve as an effective method. However, traditional CBF approaches are severely restricted in expressiveness, particularly failing to
encode nested formulas containing multiple temporal operators. While recent methods based on reachability analysis attempt to address this, they incur a heavy computational burden and rely strictly on known system dynamics. To overcome this challenge, this paper investigates a CBF-based recursive control scheme for nested STL specifications under uncertain disturbances. Within this scheme, we introduce a novel recursive CBF design procedure guided by a modified STL tree (sTLT) to yield explicit, parameterized CBFs without heavy computational demands.
To render the proposed recursive CBF design applicable to systems subject to uncertain disturbances, we further integrate a novel reconstructed CBF-based quadratic programming (QP) controller. This controller requires no prior knowledge of the disturbances while relaxing initial safety assumptions. The proposed recursive control synthesis framework is proven to effectively encode nested STL specifications while ensuring that the system satisfies the specifications under unknown disturbances, as supported by simulation results.
\end{abstract}

\section{INTRODUCTION}
As formal languages, temporal logics such as linear temporal logic (LTL) and signal temporal logic (STL) provide a comprehensive yet rigorous specification framework to constrain system behaviors~\cite{Vardi2005,Maler2004}. These formal methods have garnered widespread attention by enabling the description of highly complex mission objectives for unmanned systems, including robotics and autonomous vehicles.

In contrast to LTL, STL provides quantitative semantics for spatial-temporal properties, enabling more refined task specifications. A common approach involves encoding STL specifications as mixed-integer constraints to formulate a mixed-integer linear programming (MILP) problem within a model predictive control (MPC) framework~\cite{Raman2014}. However, such MILP-based approaches incur a significant computational burden, particularly under long time horizons or complex STL formulas.
To overcome the heavy computational burden, several frameworks have been developed for control synthesis under STL specifications~\cite{Takayama2025,Lai2025}. However, these methods either remain restricted to discrete-time systems~\cite{Takayama2025} that inherently risk inter-sample constraint violations or necessitate solving non-convex nonlinear programs~\cite{Lai2025}.

In recent years, control barrier functions (CBFs) have been introduced to solve the control synthesis problem for continuous-time systems under STL specifications~\cite{Lindemann2019a,Lindemann2020,Zhou2025}. In CBF-based approaches, STL specifications are encoded into a set of time-varying CBF constraints, the satisfaction of which guarantees the fulfillment of the STL tasks. Besides, these CBF constraints are integrated into a quadratic programming (QP) controller, yielding high computational efficiency for the CBF-QP scheme.

However, CBF-based methods are incapable of encoding arbitrary STL specifications. In the pioneering CBF-based approaches for STL specifications~\cite{Lindemann2019a}, specific CBF constructions were proposed for the \textit{Always} operator, the \textit{Eventually} operator and \textit{Conjunction}. Almost all subsequent methods have inherited this original design, rendering them incapable of encoding \textit{nested STL formulas}~\cite{Lindemann2020,Zhou2025}, in which multiple temporal operators are nested. Consequently, the applicable scope of STL formulas for existing CBF-based frameworks remains highly limited.
 
To broaden the applicable scope of CBF-based methods, several recent works have conducted in-depth explorations. Reference \cite{Buyukkocak2022} considers a richer fragment of STL specifications encompassing conflicting requirements, providing heuristics on feasible subtask ordering to facilitate CBF design. A constructive CBF approach applicable to nearly arbitrary STL formulas is proposed in~\cite{Yu2024}. By formulating an STL tree (sTLT) to guide the CBF design, this method reveals the connection between nested STL specifications and system state constraints based on system reachable sets.

The associated construction procedure in~\cite{Yu2024} relies on Hamilton-Jacobi-Bellman (HJB) reachability analysis to compute reachable sets, thereby yielding numerical CBFs for nonlinear systems. 
Nevertheless, the HJB reachability analysis relies on precise system models, rendering it difficult to construct accurate reachable sets and appropriate CBFs for uncertain systems. Since practical systems are inevitably subject to unknown disturbances, the applicability of the method in~\cite{Yu2024} is limited. Moreover, even if one attempts to rigorously incorporate disturbances into the HJB reachability analysis, such an extension would demand prior knowledge of the disturbances (e.g., their upper bounds) and require treating them as additional optimization variables, which exacerbates the inherently high computational burden~\cite{Yu2024}. These challenges drive the development of explicit parameterized CBFs~\cite{Marchesini2025} for nested STL formulas, offering a computationally lightweight solution unencumbered by the system dynamics.

Motivated by these limitations, we propose a control synthesis scheme for nested STL specifications under uncertain disturbances. Our primary contributions are summarized as follows:
\begin{itemize}
  \item In contrast to the numerical CBFs derived via reachability analysis in~\cite{Yu2024}, a novel recursive CBF design procedure based on a modified sTLT is proposed to yield explicit parameterized CBFs, which is applicable to uncertain systems. Through sliding window variables, the complex temporal relationships in nested STL formulas are captured. Besides, satisfying the CBF constraints is proven to guarantee the fulfillment of the nested STL task.
  \item To address external disturbances present within the system dynamics, a disturbance rejection QP controller based on reconstructed CBFs is proposed. In contrast to existing disturbance rejection frameworks, e.g.,~\cite{Zhou2025,Wang2023,Sun2024}, the proposed approach relaxes initial safety assumptions while requiring no prior knowledge of the disturbances.
\end{itemize}

The rest of this paper is structured as follows: Section \ref{sec:2} provides the preliminaries and describes the problem. Section \ref{sec:3} presents the proposed disturbance rejection control scheme for continuous-time systems under nested STL specifications. Section \ref{sec:4} demonstrates the effectiveness of the proposed method through a simulation example. A conclusion is drawn in Section \ref{sec:5}.

\textbf{Notation}: 
The sets of real numbers and positive integers are denoted by $\mathbb{R}$ and $\mathbb{Z}^+$, respectively. $\mathbb{R}^{n}$ denotes the $n$-dimensional real vector space and $\mathbb{R}^{m \times n}$ denotes the $m \times n$-dimensional real matrix space. The symbol $\lceil \cdot \rceil$ stands for the ceiling operator mapping a real number to the least integer greater than or equal to that number.

\section{PRELIMINARIES AND PROBLEM STATEMENT}\label{sec:2}
In this section, we provide the preliminaries on STL, the system model, and CBFs. Subsequently, the control objective of this work is presented.

\subsection{Signal Temporal Logic}
STL~\cite{Maler2004} is a specification language built upon predicates $\mu$, each of which is derived by evaluating a continuously differentiable predicate function $h_\mu: \mathbb{R}^n \rightarrow \mathbb{R}$ as
\begin{equation}
    \mu:= \begin{cases}
    \top ,& \text{if} \, \, h_\mu(x) \geq 0 \\ 
    \bot , & \text{if} \, \, h_\mu(x)<0 
\end{cases}
\end{equation}
The STL specifications considered in this work are constructed according to the following syntax:
\begin{equation}
\label{eq:STL_syntax}
  \phi::=\top|\mu|\neg\mu|\phi_1\wedge\phi_2|\mathrm{G}_{[t_a,t_b]}\phi|\phi_1\mathrm{U}_{[t_a,t_b]}\phi_2,
\end{equation}
where $\phi_1, \phi_2$ are STL formulas and $t_a$, $t_b$ are non-negative constants with $t_a\leq t_b$. Here, $\neg$ and $ \wedge$ represent logic operators \textit{Negation} and \textit{Conjunction}, while $\mathrm{G}_{[t_a,t_b]}\phi$ and $\mathrm{U}_{[t_a,t_b]}$ are temporal operators \textit{Always} and \textit{Until} defined on the interval $[t_a,t_b]$, respectively. 

The satisfaction of an STL specification $\phi$ by a continuous-time signal $x(t):\mathbb{R}^{\geq 0}\rightarrow \mathbb{R}^n$, starting from time $t$, is denoted by $(x,t)\models \phi$. The semantics of STL are given by the following recursive rules:
\begin{align}
  \label{eq:STL_sematics}
  &(x,t)\models\mu  &&\Leftrightarrow h_\mu(x(t))\geq0, \notag \\
  &(x,t)\models\neg\mu  &&\Leftrightarrow\neg((x,t)\models\mu), \notag \\
  &(x,t)\models\phi_1\wedge\phi_2  &&\Leftrightarrow(x,t)\models\phi_1\wedge(x,t)\models\phi_2, \notag \\
  &(x,t)\models \mathrm{G}_{[t_a,t_b]}\phi  &&\Leftrightarrow\forall t_{1}\in[t+t_a,t+t_b],(x,t_{1})\models\phi ,\notag \\
  &(x,t)\models\phi_1\mathrm{U}_{[t_a,t_b]}\phi_2  \! \! \! \! \! \! \! \! \! && \Leftrightarrow\left\{\exists t_1\in[t+t_a,t+t_b],(x,t_1) \models\phi_2\right\} \notag \\
  &&&\wedge\left\{\forall t_2\in[t,t_1],(x,t_2)\models\phi_1\right\} , \notag \\
  &(x,t)\models \mathrm{F}_{[t_a,t_b]}\phi  &&\Leftrightarrow\exists t_1\in[t+t_a,t+t_b],(x,t_1)\models\phi,
\end{align}
where $\mathrm{F}_{[t_a,t_b]}$ is the temporal operator \textit{Eventually} defined on the interval $[t_a,t_b]$ and $\mathrm{F}_{[t_a,t_b]}\phi$ can be represented by $\top \mathrm{U}_{[t_a,t_b]}\phi$ according to syntax (\ref{eq:STL_syntax}). 

Moreover, an STL formula is considered \textit{Nested} when its temporal operators are applied to subformulas inherently containing other temporal operators~\cite{Yu2024}, such as $\mu_1\mathrm{U}_{[t_{a,1},t_{b,1}]}(G_{[t_{a,2},t_{b,2}]}\mu_2)$ and $\mathrm{F}_{[t_{a,1},t_{b,1}]}(\mu_1 \wedge \mathrm{F}_{[t_{a,2},t_{b,2}]}\mu_2)$.

\subsection{System Model and Control Barrier Functions}
Consider a class of nonlinear systems described by:
\begin{equation}
\label{eq:sys}
        \dot x = f(x)+g(x)u + d,
\end{equation}
where $x \in \mathbb{R}^n$ is the system state, $u \in \mathbb{R}^m$ is the control input and $d \in \mathbb{R}^n$ is the unknown external disturbance. The functions $f: \mathbb{R}^n \rightarrow \mathbb{R}^n$ and $g: \mathbb{R}^m \rightarrow \mathbb{R}^{n \times m}$ are locally Lipschitz continuous.

For system (\ref{eq:sys}), consider a trajectory $z(t): \mathbb{R} \rightarrow \mathbb{R}^q$ with time derivative $\dot z$ and a continuously differentiable function $h(x,z): \mathbb{R}^n \times \mathbb{R}^q \rightarrow \mathbb{R}$. Then, the zero superlevel set of the function $h(x,z)$ is defined as the time-varying safe set $\mathcal{C}(t)$, given by
\begin{equation}
\label{eq:TV_safe_set}
        \mathcal{C}(t) = \{x(t) \in \mathbb{R}^n | h(x,z) \geq 0 \}.
\end{equation}
We further define $\mathrm{Int}(\mathcal{C}(t))$ and $\partial\mathcal{C}(t)$ as the interior and boundary of $\mathcal{C}(t)$, respectively, where 
\begin{align}
  & \mathrm{Int}(\mathcal{C}(t)):=\{x(t)\in\mathbb{R}^{n} | h(x,z)>0\}, \notag \\
  & \partial\mathcal{C}(t):=\{x(t)\in\mathbb{R}^{n} | h(x,z)=0\}.
\end{align}

Consequently, the definition of CBF is introduced below.
\begin{definition} [Zeroing CBFs, ZCBFs~\cite{Ames2017}] \!
  The continuously differentiable function $h(x,z)$ is a ZCBF if there exists an extended class $\mathcal{K}$ function $\alpha(\cdot)$ for all $x(t) \in \mathcal{C}(t)$, such that
  \begin{align}
    \label{eq:uzcbf}
    \sup_{u\in\mathbb{R}^m} & \left[\frac{\partial h}{\partial x}\left(f(x) + g(x)u + d\right) + \frac{\partial h}{\partial z} \dot z \right] \geq -\alpha(h(x,z)).
  \end{align}
\end{definition}
The set of control inputs $u$ that satisfies (\ref{eq:uzcbf}) is defined as $\mathcal{U}_{Z}$. Then, the following lemma is presented for ZCBFs:
\begin{lemma}[\cite{Ames2017}] 
  \label{lem:1} 
  For a given ZCBF $h(x,z)$ and system (\ref{eq:sys}), if the control input satisfies $u\in\mathcal{U}_Z$, the time-varying safe set $\mathcal{C}(t)$ is forward invariant.
\end{lemma}

\subsection{Problem Statement}

Current CBF-based methods~\cite{Lindemann2019a,Lindemann2020} are limited to non-nested STL formulas. While the approach in~\cite{Yu2024} accommodates nested STL specifications, the reachability analysis technique it employs is difficult to apply to the uncertain system (\ref{eq:sys}) considered herein.

Based on the limitations observed above, our control objective is to synthesize a controller for the system (\ref{eq:sys}) such that it satisfies the nested STL specifications in (\ref{eq:STL_syntax}). The specific objectives are summarized as follows:
\begin{itemize}
        \item Propose a design procedure for the CBF $h_{\phi}(x,t)$ corresponding to a given STL formula $\phi$, such that if the time-varying safe set $\mathcal{C}(t)$ associated with $h_{\phi}(x,t)$ is forward invariant over $[0,T)$, then $(x,0) \models \phi$, where $T$ is the terminal time determined by $\phi$;
        \item Synthesize a disturbance rejection controller for system (\ref{eq:sys}) to guarantee the forward invariance of the time-varying safe set $\mathcal{C}(t)$.
\end{itemize}

Achieving the control objective necessitates the following assumptions.

\begin{assumption}
\label{ass:1}
    The external disturbance $d$ in (\ref{eq:sys}) is bounded, i.e., $||d|| \leq D$, where $D$ is an \textit{unknown} positive constant.
\end{assumption}
\begin{assumption}
  \label{ass:2}
  The derivatives of all predicate functions $h_\mu(x)$, i.e., $\frac{\text{d} h_{\mu}(x)}{\text{d} x}$, are globally Lipschitz continuous.
\end{assumption}

The following lemmas present fundamental results that will be utilized to develop our subsequent methods.
\begin{lemma}[\cite{Gilpin2021}]
\label{lem:2}
        Consider a conjunction of $p$ CBFs $h_{\psi_k}(x,t)$, where $\psi_k$ $(k = 1,\dots,p)$ are STL formulas. It follows that
        \begin{align}
          \widehat {\min} \left(h_{\psi_1},\dots,h_{\psi_p} | \kappa\right)  & = -\frac{1}{\kappa} \ln\left(\sum_{k=1}^{p}\exp\left(-\kappa h_{\psi_k}\right) \right)  \notag \\
          & \leq \min \left(h_{\psi_1},\dots,h_{\psi_p}\right), 
        \end{align}
        \begin{equation}
          \lim_{\kappa \to +\infty} \widehat {\min} \left(h_{\psi_1},\dots,h_{\psi_p} | \kappa\right) = \min \left(h_{\psi_1},\dots,h_{\psi_p}\right),
        \end{equation}
where $\kappa$ is a positive constant.
\end{lemma}

\section{MAIN RESULTS} \label{sec:3}
In this section, we present a design procedure of CBFs for nested STL formulas. Simultaneously, a robust controller is synthesized based on a CBF-QP framework.

\subsection{Recursive Design of CBFs along Modified sTLT}
Inspired by~\cite{Takayama2025,Yu2024}, we construct a modified sTLT for the given STL specifications. The specific structure of this tree is then utilized to guide the design of the corresponding CBF. 


\begin{definition}[Modified sTLT] 
\label{def:3}
  A modified sTLT is defined as a directed tree $\mathcal{T}_{\phi} := \left(\mathcal{N}_{\phi}, \mathcal{E}_{\phi}\right)$ satisfying the following properties:
  \begin{itemize}
    \item [i)] $\mathcal{N}_{\phi}$ denotes the set of formula nodes while $\mathcal{E}_{\phi} := \{\left(\psi_1,\psi_2\right)_{\mathfrak{O} } \mid \mathfrak{O} \in \{\wedge,\mathrm{G}_{[t_a,t_b]},\mathrm{F}_{[t_a,t_b]} \}, \psi_1,\psi_2 \in \mathcal{N}_{\phi}\}$ denotes the set of edges;
    \item [ii)] if the edge $\left(\psi_1,\psi_2\right)_{\mathrm{T}_{[t_a,t_b]}} \in \mathcal{E}_{\phi}$ with $\mathrm{T}_{[t_{a},t_{b}]} \in \{\mathrm{G}_{[t_{a},t_{b}]}, \mathrm{F}_{[t_{a},t_{b}]}\}$, then $\psi_1 = \mathrm{T}_{[t_a,t_b]}\psi_2$;
    \item [iii)] if the edges $\left(\psi,\psi_1\right)_{\wedge},\dots, \left(\psi,\psi_p\right)_{\wedge} \in \mathcal{E}_{\phi}$, then $\psi = \bigwedge _{k=1}^p \psi_k $;
    \item [iv)] the in-degree of every node in $\mathcal{T_{\phi}}$ is at most $1$, and for any edge $\left(\psi_1,\psi_2\right)_{\mathfrak{O} } \in \mathcal{E}_{\phi}$, $\psi_1$ is defined as the parent node of $\psi_2$ while $\psi_2$ is the child node of $\psi_1$.
  \end{itemize}
  Consequently, the root node of $\mathcal{T}_{\phi}$ inherently maps to a specific STL formula $\phi$.
\end{definition}

In contrast to the sTLT in~\cite{Yu2024} which defines nodes via system reachable sets, the modified sTLT relies entirely on pure STL formulas, with its hierarchical structure subsequently guiding the explicit CBF design.
Note that condition iv) in Definition \ref{def:3} does not hinder the construction of $\mathcal{T_{\phi}}$ for a given STL formula $\phi$, since it can always be met by introducing redundant elements into $\mathcal{N}_{\phi}$. Moreover, we can deduce from Definition \ref{def:3} that the leaves of $\mathcal{T_{\phi}}$ are predicates in $\phi$. An example of a modified sTLT construction is shown in Fig. \ref{fig:sTLTa}.

As implied by Definition \ref{def:3}, a modified sTLT cannot be directly constructed for all STL formulas given in \eqref{eq:STL_syntax}. To enable the construction of a modified sTLT for a given STL formula while simultaneously simplifying its resulting structure (i.e., reducing the number of nodes or edges to facilitate the subsequent CBF design), the original formula needs to be transformed into a desired form as follows:

\begin{definition}[STL formula in desired form]
  A given STL formula $\phi$ is said to be in the desired form if:
  \begin{itemize}
    \item [i)] it contains no \textit{Until} operators;
    \item [ii)] its logical operators do not operate directly on predicates;
    \item [iii)] it contains no subformulas $\mathrm{G}_{[t_a,t_b]}\left(\psi_1 \wedge \psi_2 \right) $ and $\mathrm{G}_{[t_a,t_b]}\mu_1 \wedge \mathrm{G}_{[t_a,t_b]}\mu_2$, where $\psi_1, \psi_2$ are subformulas and $\mu_1, \mu_2$ are predicates;
    \item [iv)] it contains no subformulas $\mathrm{T}_{[t_{a,1},t_{b,1}]}\mathrm{T}_{[t_{a,2},t_{b,2}]}\psi$ wherein the successive operators $\mathrm{T}$ are identical, while $\psi$ represents a subformula.

  \end{itemize}
  
\end{definition}

To transform an STL formula into the desired form, the specific transformation rules are described as follows.

\begin{rules}
  \label{rul:1}
  For $\psi_1\mathrm{U}_{[t_a,t_b]}\psi_2$ contained in $\phi$, we rewrite it into $\mathrm{G}_{[0,t_b]}\psi_1 \wedge \mathrm{F}_{[t_a,t_b]}\psi_2$.
\end{rules}

\begin{rules}
  \label{rul:2}
  For $\lnot \mu$ and $\bigwedge^p_{k=1} \mu_k$ contained in $\phi$, where $\mu$ and $\mu_k$ $(k=1,\dots,p)$ are predicates with predicate functions $h_\mu$ and $h_{\mu_k}$, respectively, we replace them with $\mu_\lnot$ and $\mu_\wedge$, respectively. The predicate functions corresponding to $\mu_\lnot$ and $\mu_\wedge$ are given by $h_{\mu_\lnot}:=-h_\mu$ and $h_{\mu_\wedge}:=\widehat {\min}\left(h_{\mu_1},\dots,h_{\mu_p}| \kappa \right) $, respectively.
\end{rules}

\begin{rules}
  \label{rul:3}
  For $\phi_1:=\mathrm{G}_{[t_a,t_b]}\left(\psi_1 \wedge \psi_2 \right) $ and $\phi_2:=\mathrm{G}_{[t_a,t_b]}\psi_1 \wedge \mathrm{G}_{[t_a,t_b]}\psi_2$ contained in $\phi$, $\phi_1$ is equivalent to $\phi_2$.
\end{rules}

\begin{rules}
  \label{rul:4}
  For $\mathrm{T}_{[t_{a,1},t_{b,1}]}\mathrm{T}_{[t_{a,2},t_{b,2}]}\psi$ contained in $\phi$ featuring identical consecutive temporal operators, we rewrite it into $\mathrm{T}_{[t_{a,1}+t_{a,2},t_{b,1}+t_{b,2}]}\psi$.
\end{rules}

By applying the transformation Rule \ref{rul:1}-\ref{rul:4}, the given formula $\phi$ is transformed into the desired form. To alleviate the conservatism introduced by translating STL formulas into CBFs, an additional transformation rule is appended below as demonstrated in~\cite{Marchesini2025}.

\begin{rules}\label{rul:5}
  For $\mathrm{G}_{[t_{a,1},t_{b,1}]}\mathrm{F}_{[t_{a,2},t_{b,2}]}\psi$ in $\phi$, if $t_{a,2} \neq t_{b,2}$, we reformulate it as $\bigwedge_{i=1}^{p_f}\mathrm{F}_{[w_i,w_i]}\psi$, where
  \begin{align}
    \label{eq:GF1}
    &w_{i} = w_{i-1} + \delta_i (t_{b,2}-t_{a,2}), \, i=1,\dots,p_f, \\
    &w_0 = t_{a,1} + t_{a,2},
  \end{align}
  with $p_f\geq \left\lceil \frac{t_{b,1}-t_{a,1}}{t_{b,2}-t_{a,2}}\right\rceil$ and $\delta_i \in [\frac{t_{b,1}-t_{a,1}}{p_f(t_{b,2}-t_{a,2})},1]$.
\end{rules}

\begin{proposition}[\cite{Marchesini2025}] \label{pro:6}
  Let STL formula $\phi := \mathrm{G}_{[t_{a,1},t_{b,1}]}\mathrm{F}_{[t_{a,2},t_{b,2}]}\psi $ with $t_{a,2} \neq t_{b,2}$. Using Rule \ref{rul:5}, $\phi$ is transformed into $\varphi:=\bigwedge_{i=1}^{p_f}\mathrm{F}_{[w_i,w_i]}\psi$. If $(x,t)\models \varphi$ holds, then $(x,t)\models \phi$.
\end{proposition}

\begin{proposition}
  \label{pro:1}
  For a given STL specification $\phi$ and the transformed specification $\hat \phi$ obtained via Rule \ref{rul:1}-\ref{rul:5}, if $(x,0)\models \hat \phi$ holds, then $(x,0)\models \phi$.
\end{proposition}
\begin{proof}
  Proposition \ref{pro:1} establishes the soundness of the transformed specification $\hat \phi$. Specifically, the validity of Rules \ref{rul:1}, \ref{rul:3} and \ref{rul:4} follows directly from the STL semantics (\ref{eq:STL_sematics}). Furthermore, Lemma \ref{lem:2} validates Rule \ref{rul:2} while Proposition \ref{pro:6} guarantees Rule \ref{rul:5}. Consequently, the condition $(x,0)\models \hat \phi$ sufficiently implies $(x,0)\models \phi$.
\end{proof}

\begin{figure}[!htb]
\centering
\subfloat[]{\includegraphics[width=.404737\linewidth]{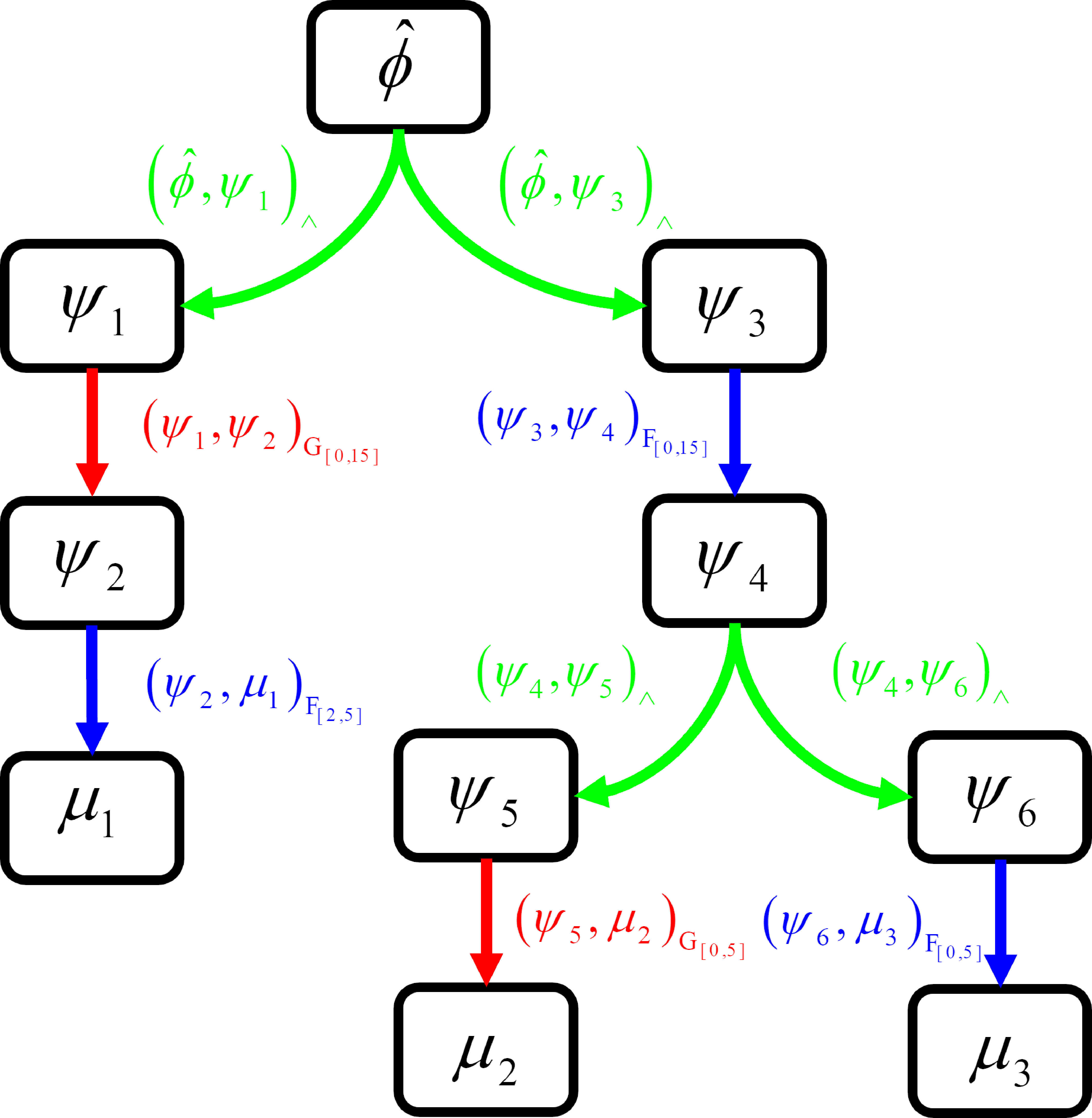}\label{fig:sTLTa}}\hspace{5pt}
\subfloat[]{\includegraphics[width=.555\linewidth]{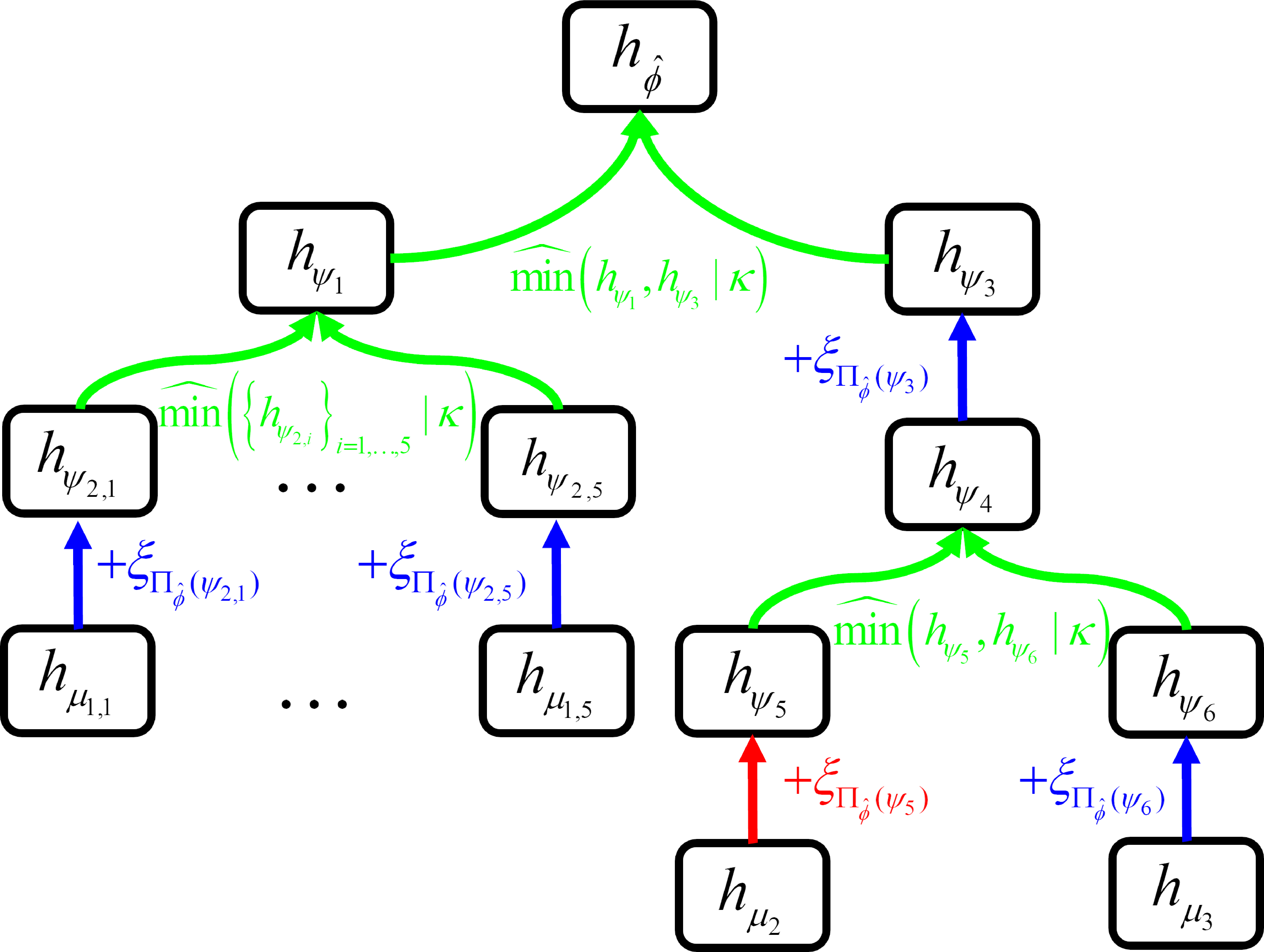}\label{fig:sTLTb}}\\
\caption{(a) Modified sTLT for $\phi = \mathrm{G}_{[0,15]} \mathrm{F}_{[2,5]}\mu_1 \wedge \mathrm{F}_{[0,15]} \left( \mathrm{G}_{[0,5]}\mu_2 \wedge \mathrm{F}_{[0,5]}\mu_3\right) $, where $\psi_1 = \mathrm{G}_{[0,15]} \psi_2$, $\psi_2 = \mathrm{F}_{[2,5]}\mu_1$, $\psi_3 = \mathrm{F}_{[0,15]} \psi_4$, $\psi_4 = \psi_5 \wedge \psi_6$, $\psi_5 = \mathrm{G}_{[0,5]}\mu_2$ and $\psi_6 = \mathrm{F}_{[0,5]}\mu_3$. (b) The backward CBF design along the edges of $\mathcal{T}_{\hat \phi}$, where $\hat \phi$ is derived by transforming the subformula $\psi_1=\mathrm{G}_{[0,15]} \mathrm{F}_{[2,5]}\mu_1$ contained in $\phi$ via Rule \ref{rul:5}.}
\label{fig:sTLT}
\end{figure}

Given an STL formula $\hat \phi$ transformed via Rules \ref{rul:1}-\ref{rul:5}, we construct the modified sTLT $\mathcal{T_{\hat \phi}}$ and define the set of formula nodes as $\mathcal{N}_{\hat \phi}: = \{ \hat \phi, \psi_1,\dots,\psi_M, \mu_1,\dots,\mu_N\}$.
Next, we design the CBF $h_{\hat \phi}$ for $\hat \phi$ along the tree $\mathcal{T}_{\hat \phi}$ in a recursive manner. The specific design procedure is as follows.

\textbf{\emph{Step 1.}} \textbf{Initialization}: Let $h_{\hat \phi},h_{\psi_1},\dots,h_{\psi_M},h_{\mu_1},\dots,$ $h_{\mu_N}$ denote the CBFs corresponding to $\hat \phi, \psi_1,\dots,\psi_M,$ $\mu_1,\dots,\mu_N$, respectively, where $h_{\mu_1},\dots,h_{\mu_N}$ are known predicate functions, while the remaining CBFs will be derived subsequently. 

Moreover, for each node in $\mathcal{N}_{\hat \phi}$, we define an active time $ \overline{t}_k$ and a terminal time $\underline{t}_k$ for $k \in \mathcal{I}_{\hat \phi}$, where $\mathcal{I}_{\hat \phi} \subseteq \mathbb{Z}^+$ is the index set of the nodes in $\mathcal{N}_{\hat \phi}$. We also introduce mappings 
\begin{align*}
  \Pi_{\hat \phi} (\cdot) : \mathcal{N}_{\hat \phi} \rightarrow \mathcal{I}_{\hat \phi}, \,\,
  \Xi _{\hat \phi} (\cdot) : \mathcal{N}_{\hat \phi} \rightarrow \mathcal{N}_{\hat \phi},
\end{align*}
where $\Pi_{\hat \phi} (\cdot)$ denotes the index mapping such that $k = \Pi_{\hat \phi} (\psi)$ is the index of $\psi$ and $\Xi_{\hat \phi} (\cdot)$ denotes the parent node mapping such that if $\left(\psi_1,\psi_2\right)_{\mathfrak{O} } \in \mathcal{E}_{\hat \phi}$, $\psi_1 = \Xi _{\hat \phi} (\psi_2)$. Furthermore, the active time for the root node $\hat \phi$ is initialized to $\overline{t}_{\Pi_{\hat \phi} (\hat \phi)} = 0$.

\textbf{\emph{Step 2.}} \textbf{Recursive Design}: Subsequently, starting from the root node $\hat \phi$, we recursively design the CBFs along the edges of the tree $\mathcal{T}_{\hat \phi}$. The following linear function is first introduced:
\begin{align}
  \label{eq:xi}
  \xi( \sigma_k| \vartheta_k) = - a_k \sigma_k + b_k, \,  k \in \mathcal{I}_{\hat \phi},
\end{align}
where $\vartheta_k:=[a_k, b_k]^T$ is a parameter to be designed and $\sigma_k$ is a sliding window variable to be determined. The design of the CBF varies depending on the type of edge. The specific design procedures are detailed as follows:

\noindent\textbf{\emph{Case I.}} $\boldsymbol{(\psi_i, \psi_j)_{\mathrm{T}_{[t_a,t_b]}}}$: In this case, it follows that $\psi_i = \mathrm{T}_{[t_a,t_b]} \psi_j$, the CBF is designed as follows:
\begin{subequations}
\label{eq:always_CBF}
\begin{align}
   h_{\psi_i}(x,t)&:= h_{\psi_j}(x,t) + \xi(\sigma_{\Pi_{\hat \phi}(\psi_i)}| \vartheta_{\Pi_{\hat \phi}(\psi_i)}), \\
   \sigma_{\Pi_{\hat \phi}(\psi_i)} &:= 
    \begin{cases}
    0 , & \text{if} \,  t < \overline{t}_{\Pi_{\hat \phi}(\psi_i)}  \\ 
    t - \overline{t}_{\Pi_{\hat \phi}(\psi_i)} , & \text{if} \,   \overline{t}_{\Pi_{\hat \phi}(\psi_i)} \leq t < \underline{t}_{\Pi_{\hat \phi}(\psi_i)}  \\
    \underline{t}_{\Pi_{\hat \phi}(\psi_i)} - \overline{t}_{\Pi_{\hat \phi}(\psi_i)} , & \text{if} \,  t \geq \underline{t}_{\Pi_{\hat \phi}(\psi_i)}
    \end{cases}\\
  \text{s.t. } \qquad & \xi(\sigma_{\Pi_{\hat \phi}(\psi_i)}(\underline{t}_{\Pi_{\hat \phi}(\psi_i)})| \vartheta_{\Pi_{\hat \phi}(\psi_i)}) = 0, \\
  \label{eq:time_design1}
  & \underline{t}_{\Pi_{\hat \phi}(\psi_i)} := \overline{t}_{\Pi_{\hat \phi}(\psi_i)} + t^\star, \, \overline{t}_{\Pi_{\hat \phi}(\psi_i)} :=  \underline{t}_{\Pi_{\hat \phi}\left(\Xi _{\hat \phi} (\psi_i)\right) },
\end{align}
\end{subequations}
where $t^\star \in [t_a,t_b]$ if $\psi_i = \mathrm{F}_{[t_a,t_b]} \psi_j$ and $t^\star = t_a$ if $\psi_i = \mathrm{G}_{[t_a,t_b]} \psi_j$. Herein, $\sigma_{\Pi_{\hat \phi}(\psi_i)}$ is referred to as a sliding window variable.

\noindent\textbf{\emph{Case II.}} $\left\{\boldsymbol{(\psi_{k_0}, \psi_{k_1})_{\wedge}},\dots,\boldsymbol{(\psi_{k_0}, \psi_{k_p})_{\wedge}}\right\} $: In this case, we have $\psi_{k_0} = \bigwedge_{l=1}^ p \psi_{k_l}$, the CBF is designed as:
\begin{subequations}
\label{eq:conjunction_CBF}
\begin{align}
   h_{\psi_{k_0}}(x,t)&:= \widehat {\min}\left(h_{\psi_{k_1}}(x,t),\dots,h_{\psi_{k_p}}(x,t)| \kappa\right) , \\
  \text{s.t. } \qquad 
  \label{eq:time_design3}
  & \underline{t}_{\Pi_{\hat \phi}(\psi_{k_0})} := \overline{t}_{\Pi_{\hat \phi}(\psi_{k_0})} , \, \overline{t}_{\Pi_{\hat \phi}(\psi_{k_0})} :=  \underline{t}_{\Pi_{\hat \phi}\left(\Xi _{\hat \phi} (\psi_{k_0})\right) } .
\end{align}
\end{subequations}

Moreover, we set $\overline t_{\Pi_{\hat \phi}\left(\mu_k\right)} =\underline t_{\Pi_{\hat \phi}\left(\mu_k\right)} = \underline{t}_{\Pi_{\hat \phi}\left(\Xi _{\hat \phi} (\mu_k)\right) }$ for any predicate $\mu_k \in \mathcal{N}_{\hat \phi}$.

\textbf{\emph{Step 3.}} \textbf{Assembly of $h_{\hat \phi}$}: From Step 2, starting from the root node $\hat \phi$, the active time $\overline{t}_{\Pi_{\hat \phi}(\psi_k)}$ and terminal time $\underline{t}_{\Pi_{\hat \phi}(\psi_k)}$ can be obtained for all nodes $\psi_k \in \mathcal{N}_{\hat \phi}$. Thus, the sliding window variable $\sigma_{\Pi_{\hat \phi}(\psi_k)}$ and the linear function $\xi_{\Pi_{\hat \phi}(\psi_k)} := \xi(\sigma_{\Pi_{\hat \phi}(\psi_k)}|\vartheta_{\Pi_{\hat \phi}(\psi_k)})$ can be designed sequentially for the relevant nodes.
Then, proceeding from the leaf (predicate) nodes $\mu_k \in \mathcal{N}_{\hat \phi}$ with known CBFs $h_{\mu_k}$, we assemble the CBFs for other nodes along the edges of $\mathcal{T}_{\hat \phi}$ based on (\ref{eq:always_CBF})-(\ref{eq:conjunction_CBF}), ultimately yielding the CBF $h_{\hat \phi}(x,t)$ for the root node $\hat \phi$. The backward design procedure along the edges in Step 3 is illustrated by an example in Fig. \ref{fig:sTLTb}.

Then, one of the main theorems of this paper is established.
\begin{theorem}
  \label{the:1} 
  Consider an STL specification $\phi$ defined in (\ref{eq:STL_syntax}) and the transformed specification $\hat \phi$. Let $h _{\hat \phi}$ be the CBF obtained via the recursive design (\ref{eq:always_CBF})-(\ref{eq:conjunction_CBF}). If $h _{\hat \phi}(x,t) \geq 0$ holds $\forall t \in [0,T)$ with a sufficiently large horizon $T$, then $(x,0)\models \phi$. 
\end{theorem}
\begin{proof}
By the recursive definitions in (\ref{eq:time_design1})-(\ref{eq:time_design3}), the active time of a child node equals the terminal time of its parent, i.e., $\overline{t}_{\Pi_{\hat \phi}(\psi_{q})} = \underline{t}_{\Pi_{\hat \phi}\left(\Xi _{\hat \phi} (\psi_{q})\right) }$ for any $\psi_q \in \mathcal{N}_{\hat \phi}$. Furthermore, since the piecewise variable $\sigma_{\Pi_{\hat \phi}(\psi_q)}$ is monotonically increasing with respect to time, the linear function $\xi_{\Pi_{\hat \phi}(\psi_q)}$ is monotonically decreasing. Given that $\xi(\sigma_{\Pi_{\hat \phi}(\psi_q)}(\underline{t}_{\Pi_{\hat \phi}(\psi_q)})| \vartheta_{\Pi_{\hat \phi}(\psi_q)}) = 0$, the condition $h_{\hat \phi}(x,t) \geq 0$ over $[0,T)$ guarantees the non-negativity of all intermediate CBFs. Specifically, for any node $\psi_q \in \mathcal{N}_{\hat \phi}$, we have:
\begin{align}
  \label{eq:17}
  h_{\psi_q}(x,t) \geq 0, \quad \forall t \geq \overline{t}_{\Pi_{\hat \phi}(\psi_{q})} = \underline{t}_{\Pi_{\hat \phi}\left(\Xi _{\hat \phi} (\psi_{q})\right) }. 
\end{align}

We proceed by mathematical induction on the structure of the modified sTLT to establish that $(x,\overline{t}_{\Pi_{\hat \phi}(\psi_{q})}) \models \psi_q$ for all $\psi_q \in \mathcal{N}_{\hat \phi}$. 

\noindent \textbf{Base Case:} For any predicate node $\mu_k \in \mathcal{N}_{\hat \phi}$, (\ref{eq:17}) implies $h_{\mu_k}(x, t) \geq 0$ for $t \in [\overline{t}_{\Pi_{\hat \phi}(\mu_{k})},T)$. By definition, this directly yields $(x,\overline{t}_{\Pi_{\hat \phi}(\mu_{k})}) \models \mu_k$.

\noindent \textbf{Inductive Step:} Assume the inductive hypothesis $(x,\underline{t}_{\Pi_{\hat \phi}(\psi_k)}) \models \psi_k$ holds for a child node $\psi_k$. We evaluate its parent node $\psi_l$:
\begin{enumerate}
    \item \textit{For $\psi_l = \mathrm{F}_{[t_a,t_b]}\psi_k$:} The recursive construction yields $\overline{t}_{\Pi_{\hat \phi}(\psi_{k})} = \overline{t}_{\Pi_{\hat \phi}(\psi_l)} + t^\star$ for $t^\star \in [t_a, t_b]$. The inductive hypothesis guarantees satisfaction at this specific instant $t^\star$, implying $(x,\overline{t}_{\Pi_{\hat \phi}(\psi_l)}) \models \mathrm{F}_{[t_a,t_b]} \psi_k$ according to the semantics in (\ref{eq:STL_sematics}).
    
    \item \textit{For $\psi_l = \mathrm{G}_{[t_a,t_b]}\psi_k$:} Let $\hat \phi'$ be the formula derived from $\hat \phi$ by explicitly replacing its subformula $\psi_l = \mathrm{G}_{[t_a,t_b]} \psi_k$ with $\psi_l = \mathrm{G}_{[t_a+\delta,t_b]} \psi_k$, where $\delta \in [0,t_b-t_a]$ represents a temporal shift. It is evident that the corresponding modified sTLTs $\mathcal{T}_{\hat \phi}$ and $\mathcal{T}_{\hat \phi'}$ share an identical structure. Similarly, we recursively design the CBF $h_{\hat \phi'}(x,t)$ and deliberately enforce identical parameters $b_{\Pi_{\hat \phi}(\psi_q)} = b_{\Pi_{\hat \phi'}(\psi_q)}$ for all corresponding nodes $\psi_q$ across both trees. Due to this identical tree structure and the shared parameters, the monotonic properties dictate that $h_{\hat \phi'}(x,t) \geq h_{\hat \phi}(x,t) \geq 0$. The inductive hypothesis guarantees $(x,\overline{t}_{\Pi_{\hat \phi}(\psi_l)} + t_a + \delta) \models \psi_k$. Since this satisfaction holds for any arbitrary shift $\delta \in [0,t_b-t_a]$, it directly follows from the STL semantics that $(x,\overline{t}_{\Pi_{\hat \phi}(\psi_l)}) \models \mathrm{G}_{[t_a,t_b]} \psi_k = \psi_l$.
    
    \item \textit{For Conjunction $\wedge$:} For $\psi = \bigwedge_{k=1}^p \psi_k$, the active times synchronize $\overline{t}_{\Pi_{\hat \phi}(\psi)} = \overline{t}_{\Pi_{\hat \phi}(\psi_k)}$. The under-approximation of the minimum operator ensures all subformulas are satisfied simultaneously, yielding $(x, \overline{t}_{\Pi_{\hat \phi}(\psi)}) \models \psi$.
\end{enumerate}

By induction, $(x,\overline{t}_{\Pi_{\hat \phi}(\hat \phi)} = 0) \models \hat \phi$ holds. Combining this result with Proposition \ref{pro:1} concludes the proof, demonstrating that $(x, 0) \models \phi$.
\end{proof}

\begin{remark} \label{re:1}
  Since the complete path from the root node $ \hat \phi$ to a given predicate node $\mu_k$ is unique, the release time $T_{\hat \phi}(\mu_k)$ is given by the sum of the interval endpoints $t_b$ of all temporal operators $\mathrm{T}_{[t_a,t_b]}$ on this path. Inspired by~\cite{Lindemann2019a}, to prevent expired STL specifications from affecting the execution of subsequent tasks, we set $h_{\mu_k} = +\infty$ once $t \geq T_{\hat \phi}(\mu_k)$, thereby rendering it ineffective on the overall CBF $h_{\hat{\phi}}$. Additionally, $T = \max_{\mu_k \in \mathcal{N}_{\hat \phi}}(T_{\hat \phi}(\mu_k))$ can be established in Theorem \ref{the:1} since any subsequent system behavior beyond this time becomes irrelevant to the satisfaction of the STL specification.
\end{remark}

\subsection{Disturbance Rejection Control using Reconstructed CBFs}
Building upon the CBF $h_{\hat \phi}$ recursively designed for nested STL formulas, a CBF-QP can render the corresponding safe set $\mathcal{C}(t):=\left\{x \in \mathbb{R}^n \mid  h_{\hat \phi}(x,t)\geq 0\right\} $ forward invariant to satisfy the STL specifications. However, the presence of unknown disturbances in the system renders the standard CBF constraint \eqref{eq:uzcbf} inapplicable. To address the unknown disturbances, a disturbance rejection controller is synthesized based on the reconstructed CBFs.

Note that due to the piecewise nature of function \eqref{eq:xi} in the recursive design and the presence of the timeout pruning mechanism in Remark \ref{re:1}, $h_{\hat \phi}$ is piecewise continuously differentiable. Let us define the switching time sequence of the CBF as $\mathcal{S}_{\hat \phi}=\{s_0=0, s_1, \dots, s_q\}$, where $q$ denotes the total number of switches. 
This sequence encompasses all $\overline{t}_k$, $\underline{t}_k$ and $T_{\hat \phi}(\mu_k)$, ensuring that $h_{\hat{\phi}}$ is continuously differentiable over any $[s_{j-1}, s_j)$ for $j = 1,\dots,q$.
Besides, $$\tau(t) := \max \left\{ T_{\hat \phi}(\mu_k) \mid \mu_k \in \mathcal{N}_{\hat \phi}, \, T_{\hat \phi}(\mu_k) \leq t \right\}$$ is defined as the last release time.

To proceed with the subsequent content, the following assumption and lemma are required.
\begin{assumption}
  \label{ass:3} 
  At the initial time $t=0$, the system is in the interior of the safe set $ \mathcal{C}(0)$, i.e., $h_{\hat \phi}(x,0) >0$.
\end{assumption}
\begin{lemma}[\cite{Li2022}]
  \label{lem:3} 
    The following inequality holds for any $\epsilon>0$ and any $\varpi \in \mathbb{R}$:
    \begin{equation}
      |\varpi| \leq \frac{\varpi^2}{\sqrt{\varpi^2 + \epsilon^2}} + \epsilon.
    \end{equation}
\end{lemma}

\begin{remark}
Compared to existing literature, the assumptions introduced in this paper are relatively milder. Specifically, Assumption \ref{ass:1} dictates that the proposed method requires merely the boundedness of disturbances while demanding no exact prior knowledge, directly contrasting with the reliance on known disturbance upper bounds seen in~\cite{Wang2023}. Furthermore, regarding the initial safety condition outlined in Assumption \ref{ass:3}, existing studies~\cite{Zhou2025,Sun2024} frequently postulate $h_{\hat \phi}(x,0) > \beta V_d(0) \geq 0$, where $\beta$ is a positive constant and $V_d(0) \geq 0$ acts as an initial Lyapunov function value related to the disturbance estimation error. In practice, $V_d(0)$ is a positive constant that cannot be predetermined, potentially rendering the condition $h_{\hat \phi}(x,0) > \beta V_d(0)$ overly stringent. In contrast, Assumption \ref{ass:3} allows the initial state $x(0)$ to be positioned much closer to the safety boundary, thereby facilitating practical initial configurations.
\end{remark}

To begin the synthesis, the following reference model is introduced:
\begin{equation}
\dot{\hat x } = f(x) + g(x)u + \lambda\left( x - \hat x\right) ,
\end{equation}
where $\hat x \in \mathbb{R}^n$ is the estimate of $x$ and $\lambda > \frac{1}{2}$ is a constant. 
Next, we introduce the following reconstructed CBF:
\begin{align}
  \label{eq:rCBF}
  & \hat h_{\hat \phi} (\hat \upsilon ,\eta) = h_{\hat \phi}(\hat \upsilon) - \eta,
\end{align}
where $\hat \upsilon: =[\hat x ^T, t]^T$, $\hat h_{\hat \phi}$ is the reconstructed CBF, $h_{\hat \phi}(\hat \upsilon)$ is obtained by substituting $x$ with $\hat{x}$ in $h_{\hat \phi}(x,t)$, while $\eta$ acts as an adaptive parameter. 
Let
\begin{equation*} 
  e := h_{\hat \phi}(\upsilon) - \hat h_{\hat \phi} (\hat \upsilon ,\eta)
\end{equation*}
denote the reconstruction error, where $\upsilon: =[x ^T, t]^T$. 
Furthermore, by introducing an additional adaptive parameter $\hat r$, the update laws for both parameters are designed as follows:
\begin{align} \label{eq:eta_dot}
   & \dot \eta = \begin{aligned}[t] 
      & -c \frac{e(\rho - e)}{\rho}\varepsilon - \frac{\varrho e}{\rho}(\rho_0 - \rho_\infty) \mathrm{e}^{-\varrho (t - \tau(t))} \\
      & - \frac{\rho \varepsilon}{4e(\rho - e)} - \frac{\hat r^2 \left(\left\lVert \frac{\text{d} \hat h_{\hat \phi}(\hat \upsilon)}{\text{d} \hat \upsilon}\right\rVert + \left\lVert  \dot{\bar z} \right\rVert \right) \chi}{\sqrt{\chi^2\hat r^2 + \epsilon^2}},
    \end{aligned} \notag \\
  & \dot {\hat r} = \frac{ \gamma|\varepsilon| \rho}{2e(\rho-e)}\left(\left\lVert \frac{\text{d} \hat h_{\hat \phi}(\hat \upsilon)}{\text{d} \hat \upsilon}\right\rVert + \left\lVert  \dot{\bar z} \right\rVert \right) - \varsigma \hat r, 
\end{align}
where 
\begin{align*}
  & \varepsilon = \frac{1}{2} \ln\left(\frac{e}{\rho-e}\right) ,\, \chi = \frac{\varepsilon \rho}{2e (\rho-e )}\left(\left\lVert \frac{\text{d} \hat h_{\hat \phi}(\hat \upsilon)}{\text{d} \hat \upsilon}\right\rVert + \left\lVert  \dot{\bar z} \right\rVert \right) \\
  &\rho(t) = \left(\rho_0 - \rho_{\infty}\right) \mathrm{e}^{-\varrho (t - \tau(t))} + \rho_{\infty}, 
\end{align*}
with $c$, $\varrho$, $\gamma$, $\varsigma$, $\epsilon$ and $\rho_0 > \rho_\infty$ being positive constants. Besides, $\dot {\bar z}: = [\dot z ^T, 1]^T$ and $\dot{z}$ serves as a differentiator to generate an estimate of $\dot{x}$ without explicitly incorporating the control input $u$.
Further, at the release time $T_{\hat \phi}(\mu_k)$ for $\mu_k \in \mathcal{N}_{\hat \phi}$, the parameters are required to satisfy $0 < e(T_{\hat \phi}(\mu_k)) < \rho_{\infty}$ and $\hat h_{\hat \phi} (\hat \upsilon(T_{\hat \phi}(\mu_k)) ,\eta(T_{\hat \phi}(\mu_k))) \geq 0$. 
\begin{proposition} \label{pro:7}
  Consider the reconstructed CBF \eqref{eq:rCBF} and \eqref{eq:eta_dot}, under Assumption \ref{ass:1}-\ref{ass:3}, it can be concluded that:
  \begin{itemize}
    \item [i)] The estimation error $\tilde x:=x - \hat x$ remains globally uniformly bounded with $\lim_{t \to \infty} \sup \tilde x \leq \frac{D}{\sqrt{2\lambda - 1}} $;
    \item [ii)] Assuming the estimation error $\dot {\tilde{z}}:= \dot{x} - \dot{z}$ is bounded, the reconstruction error $e$ satisfies $0< e(t) < \rho(t)$ for any $t \in [s_{j-1},s_j)$, $j=1,\dots,q$ and adaptive parameters $\eta$, $\hat r$ are bounded.
\end{itemize}
\end{proposition}
\begin{proof}
  i) By constructing the Lyapunov function
  \begin{equation}
    V_1 = \frac{1}{2}\tilde x^T \tilde x,
  \end{equation}
  and applying Young's inequality, it can be deduced that 
  \begin{equation}
    \dot V_1 \leq -(2\lambda-1) V_1 + \frac{1}{2}D^2.
  \end{equation}
  Thus, it can be obtained that $\tilde x$ remains bounded with $\lim_{t \to \infty} \sup \tilde x \leq \frac{D}{\sqrt{2\lambda - 1}} $. 
  
  ii) 
  The proof proceeds by contradiction. Suppose that there exists a finite time $\bar{t} $ at which $e$ exceeds the performance boundary $(0, \rho(t))$ for the first time within the interval $[0, s_1)$. Thus, it can be obtained that $\lim_{t \to \bar t} \varepsilon = \infty  $ and for $t \in [0,\bar t)$, $e \in (0,\rho(t))$. The following Lyapunov function is selected as 
  \begin{equation}\label{eq:V_2}
    V_2 = \frac{1}{2}\varepsilon^2 + \frac{1}{2 \gamma} \tilde r^2,
  \end{equation}
  where $\tilde r = r - \hat r$ and $r$ is a positive constant satisfying certain conditions to be introduced later. Taking the derivative of \eqref{eq:V_2} over the interval $[0, \bar t)$ yields
  \begin{align}
    \dot V_2  = & \varepsilon \left(\frac{\rho}{2 e\left(\rho - e\right) }\dot \eta - \frac{1}{2 \left(\rho - e\right) } \dot \rho\right) + \frac{\varepsilon \rho}{2 e\left(\rho - e\right)} \left[\frac{\text{d} h_{\hat \phi}(\upsilon)}{\text{d} \upsilon} \right. \notag \\
    & \times  \left. \left(\dot \upsilon - \dot {\hat \upsilon }\right) + \left(\frac{\text{d} h_{\hat \phi}(\upsilon)}{\text{d} \upsilon}  - \frac{\text{d}\hat h_{\hat \phi}(\hat \upsilon)}{\text{d} \hat \upsilon }\right)  \dot {\hat \upsilon }\right]  - \frac{1}{\gamma}\tilde r \dot{\hat r}.
  \end{align}
Combining Assumptions \ref{ass:1} and \ref{ass:2} with the first point of Proposition \ref{pro:7}, it follows that:
\begin{align}
  & \left\lVert \dot \upsilon - \dot {\hat \upsilon } \right\rVert = \left\lVert d - \lambda \tilde x \right\rVert \leq \Delta_1, \\
  & \left\lVert \frac{\text{d} h_{\hat \phi}(\upsilon)}{\text{d} \upsilon}  - \frac{\text{d}\hat h_{\hat \phi}(\hat \upsilon)}{\text{d} \hat \upsilon }\right\rVert  \leq L_2 \left\lVert \tilde x\right\rVert  \leq \Delta_2,
\end{align}
where $L_2$ is a Lipschitz constant of $\frac{\text{d} h_{\hat \phi}(\upsilon)}{\text{d} \upsilon} $, $\Delta_{1}$ and $\Delta_{2}$ are positive constants. Defining $r:= \max\{\Delta_{1},\Delta_{2}\}$ and applying Lemma \ref{lem:3}, it is derived that
\begin{align} \label{eq:dV_2 2}
  \dot V_2 \leq & \frac{\rho \varepsilon}{2 e\left(\rho - e\right) } \dot \eta - \frac{\varepsilon }{2 \left(\rho - e\right) } \dot \rho   + \frac{|\varepsilon| \rho}{2e(\rho-e)} \notag \\
  & \times \left(\left\lVert \frac{\text{d} h_{\hat \phi}(\upsilon)}{\text{d} \upsilon}\right\rVert + \left\lVert \dot {\hat \upsilon }\right\rVert - \left\lVert \frac{\text{d}\hat h_{\hat \phi}(\hat \upsilon)}{\text{d} \hat \upsilon}\right\rVert - \left\lVert \dot {\bar z} \right\rVert\right)   r  \notag \\
  & + \frac{|\varepsilon| \rho \tilde{r}}{2e(\rho-e)}\left(\left\lVert \frac{\text{d} \hat h_{\hat \phi}(\hat \upsilon)}{\text{d} \hat \upsilon}\right\rVert + \left\lVert  \dot{\bar z} \right\rVert \right) \notag   \\
  & +  \frac{\chi ^2\hat r^2 }{\sqrt{\chi^2 \hat r^2+ \epsilon^2}} + \epsilon - \frac{1}{\gamma}\tilde r \dot{\hat r}.
\end{align}
Applying the reverse triangle inequality yields:
\begin{align} \label{eq:reverse triangle inequality}
& \frac{|\varepsilon| \rho}{2e(\rho-e)}  \left(\left\lVert \frac{\text{d} h_{\hat \phi}(\upsilon)}{\text{d} \upsilon}\right\rVert + \left\lVert \dot {\hat \upsilon }\right\rVert - \left\lVert \frac{\text{d}\hat h_{\hat \phi}(\hat \upsilon)}{\text{d} \hat \upsilon}\right\rVert - \left\lVert \dot {\bar z} \right\rVert\right)   r  \notag 
\\
& \leq \frac{|\varepsilon| \rho}{2e(\rho-e)}  \left(\left\lVert \frac{\text{d} h_{\hat \phi}(\upsilon)}{\text{d} \upsilon} - \frac{\text{d}\hat h_{\hat \phi}(\hat \upsilon)}{\text{d} \hat \upsilon}\right\rVert + \left\lVert \dot {\hat \upsilon } - \dot {\bar z} \right\rVert\right) r \notag \\
& \leq \frac{|\varepsilon|^2 \rho^2}{8e^2(\rho-e)^2} + \frac{1}{2} \Delta ^2,
\end{align}
where it can be deduced that $\dot {\hat \upsilon } - \dot {\bar z}$ is also bounded, and $\Delta \geq \left(\left\lVert \frac{\text{d} h_{\hat \phi}(\upsilon)}{\text{d} \upsilon} - \frac{\text{d}\hat h_{\hat \phi}(\hat \upsilon)}{\text{d} \hat \upsilon}\right\rVert + \left\lVert \dot {\hat \upsilon } - \dot {\bar z} \right\rVert\right) r$ is a positive constant. Substituting \eqref{eq:eta_dot} and \eqref{eq:reverse triangle inequality} into \eqref{eq:dV_2 2} and applying Young's inequality yields:
\begin{align} \label{eq:dV_2 3}
  \dot V_2 \leq & -\frac{c}{2} \varepsilon^2 +  \frac{\varsigma}{\gamma} \tilde{r} \hat r + \frac{1}{2}\Delta^2 + \epsilon \notag \\
  \leq &  -\frac{c}{2} \varepsilon^2  -\frac{\varsigma}{2\gamma} \tilde{r}^2 + \frac{\varsigma}{2\gamma} {r}^2 + \frac{1}{2}\Delta^2 + \epsilon \notag \\
  \leq & - \zeta  V_2 + \Omega,
\end{align}
where $\zeta  = \min\{c,\varsigma\}$ and $\Omega: = \frac{\varsigma}{2\gamma} {r}^2 + \frac{1}{2}\Delta^2 + \epsilon$ is a positive constant. Integrating both sides of \eqref{eq:dV_2 3} yields that 
\begin{equation}
  \label{eq:V_2(t)}
  V_2(t) \leq \left(V_2(0) - \frac{\Omega}{\zeta}\right) \mathrm{e}^{-\zeta t} +  \frac{\Omega}{\zeta}.
\end{equation}
Equation (\ref{eq:V_2(t)}) implies that $V_2(\bar t^-)$ and thus $\varepsilon(\bar t^-)$ are bounded, contradicting the assumption $\lim_{t \to \bar t}\varepsilon(t) = \infty$. Hence, $\varepsilon$ and $\tilde r$ remain bounded for all $t \in [0,s_1)$, ensuring the reconstruction error $e$ strictly satisfies the prescribed performance criteria.

Then, if $s_1$ is a release time, an appropriate reselection of parameters ensures that $e(s_1) \in (0, \rho(s_1))$. Conversely, if $s_1$ is not a release time, the conditions $h_{\hat \phi}(\upsilon(s_1^-)) = h_{\hat \phi}(\upsilon(s_1))$ and $\hat h_{\hat \phi}(\hat \upsilon(s_1^-)) = \hat h_{\hat \phi}(\hat \upsilon(s_1))$ hold, keeping the reconstruction error within the prescribed performance bounds at $t=s_1$. In either scenario, repeating the aforementioned proof procedure establishes the validity of the conclusion over the subsequent interval $[s_1, s_2)$. Ultimately, this iterative reasoning proves that the conclusion holds across any arbitrary interval $[s_{j-1}, s_j)$ for $j=1,\dots ,q$. Then, as $e$ is bounded, combined with the local Lipschitz continuity of $h_{\hat{\phi}}$, it can be deduced that $\eta$ is also bounded. Since both $r$ and $\tilde{r}$ are bounded, it follows that $\hat{r}$ is bounded.
\end{proof}
\begin{remark}
  Regarding the design of the differentiator $\dot{z}$, one can simply employ the dynamics $\dot{z} = f(x) + \lambda(x - z)$. This design yields a bounded estimation error provided that the control effort $g(x)u$ remains bounded. We will implement this specific design in the subsequent simulation section. Notably, the differentiator serves as a modular component within the reconstructed CBF framework. It can be flexibly substituted with other advanced architectures, such as high-gain observers or unknown input observers. As long as the chosen estimator guarantees a bounded estimation error, the core conclusions of the reconstructed CBF remain entirely unaffected.
\end{remark}

Utilizing the reconstructed CBF, the following disturbance rejection controller is synthesized: 
\begin{equation}
\begin{aligned} \label{eq:qp}
&u = \mathop{\arg \min} \limits_{ u \in \mathbb{R}^{m}}  \frac{1}{2} u ^T Q u\\
&\begin{array}{r@{\quad}r@{}l@{\quad}l}
\text{s.t.} & \frac{\partial \hat h_{\hat \phi}}{\partial \hat x}\left(f(x) + g(x) u  + \lambda \tilde x \right) - \dot \eta + \frac{\partial \hat h_{\hat \phi}}{\partial t} \geq - \alpha(\hat h_{\hat \phi}), 
\end{array}
\end{aligned}
\end{equation}
where $Q \in \mathbb{R}^{m \times m}$ is a positive definite weight matrix.

\begin{theorem} \label{the:2}
  Consider the system \eqref{eq:sys} and a given STL specification $\phi$. Under Assumption \ref{ass:1}-\ref{ass:3}, by employing the CBF derived from the proposed recursive design, combined with the reconstructed CBF-based QP controller \eqref{eq:rCBF}, \eqref{eq:eta_dot} and \eqref{eq:qp}, $(x,0)\models \phi$ is guaranteed.
\end{theorem}
\begin{proof}
Under Assumption \ref{ass:3}, at $t=0$, it is always possible to choose the parameters such that $0 < e(0) < \rho_{0}$ and $\hat h_{\hat \phi} (\hat \upsilon(0) ,\eta(0) )\geq 0$. For instance, one can choose $\eta(0) = {h}_{\hat \phi}(\hat \upsilon(0)) - \frac{h_{\hat \phi}(\upsilon(0))}{2}$ and $\rho_{0} = h_{\hat \phi}(\upsilon(0))$. 
  According to Proposition \ref{pro:7} and Lemma \ref{lem:1}, the controller \eqref{eq:qp} renders the safe set $\mathcal{C}(t)$ forward invariant over the interval $[s_{0}, s_1)$.
  Based on the recursive design and the mechanism detailed in Remark \ref{re:1}, it follows that $h_{\hat \phi}(x,s_1) \geq h_{\hat \phi}(x,s_1^-) \geq 0$ holds strictly. 
If $s_1$ is not a release time, the parameters remain unchanged; otherwise, suitable parameters are reselected. In either case, $h_{\hat \phi}(\upsilon(s_1)) > \hat h_{\hat \phi}(\hat \upsilon(s_1),\eta(s_1)) \geq 0$ holds with the reconstruction error maintained within the performance bounds. Repeating this proof establishes the forward invariance of $\mathcal{C}(t)$ over any interval $[s_{j-1}, s_j)$.
Then, based on Theorem \ref{the:1}, Theorem \ref{the:2} can be proven.
\end{proof}

\section{SIMULATION} \label{sec:4}
Consider the following mobile robot model:
\begin{align} \label{eq:rb}
  & \dot{\boldsymbol{p}} = g(\theta)u + d, \, \dot \theta = \omega,
\end{align}
where ${\boldsymbol{p}} = [p_x, p_y]^T$ denotes the position of the robot, and $\theta$ is the heading angle. The matrix $g(\theta)$ is defined as $g(\theta) = \begin{bmatrix} \cos\theta & -l\sin\theta \\ \sin\theta & l\cos\theta \end{bmatrix}$, with $l = 0.036$ being a model parameter. The control input is defined as $u = [v, \omega]^T$, where $v$ and $\omega$ represent the linear and angular velocities, respectively. Additionally, velocity constraints $|v| \leq 2 $ and $|\omega| \leq 3 $ are imposed. The disturbance is set to $d = [-0.1, 0.2 \cos(0.5 t)]^T$.

We define the following nested STL task:
\begin{align} \label{eq:STL_complex}
  \phi: = &\mathrm{G}_{[0,10]}\mathrm{F}_{[0,5]} \mu_3 \wedge \mathrm{F}_{[5,6]}\mathrm{G}_{[1,2]} \mu_2 \wedge \mathrm{F}_{[12,13]}\left(\mu_3 \mathrm{U}_{[1,2]} \mu_1\right) \notag \\
  & \wedge \mathrm{G}_{[0,24]} \neg \mu_4 \wedge \mathrm{F}_{[0,24]} \mu_5.
\end{align}
The regions satisfying the predicate functions $h_{\mu_k} \ge 0$ for $k=1,\dots,5$ are depicted in Fig.~\ref{fig:2}. Based on the proposed method, we construct the following CBF:

\begin{figure}[!htb]
\centering
\includegraphics[width=8.87cm]{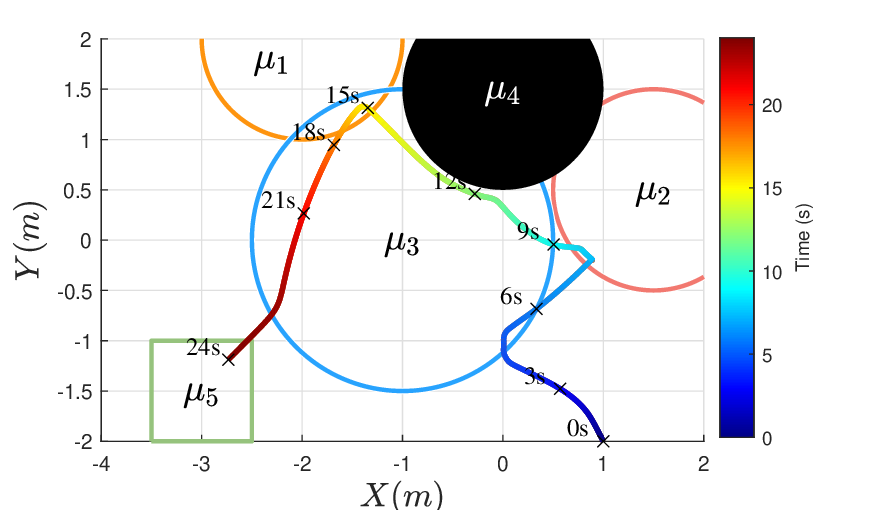}
\caption{The trajectory of a mobile robot with dynamics \eqref{eq:rb} under STL specification \eqref{eq:STL_complex}.}
\label{fig:2}
\end{figure}
\begin{figure}[!htb]
\centering
\subfloat[Reconstructed CBF $\hat h_{\hat \phi}$.]{\includegraphics[width=.48\linewidth]{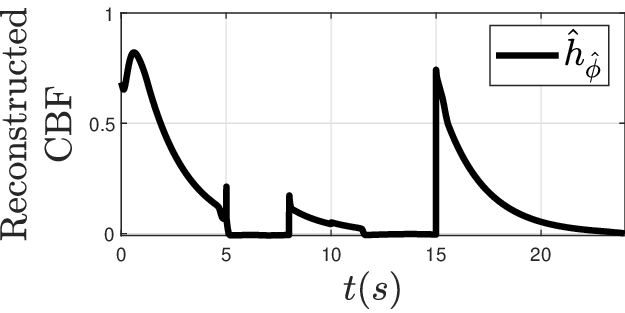}}\hspace{5pt}
\subfloat[Reconstruction error $e$.]{\includegraphics[width=.48\linewidth]{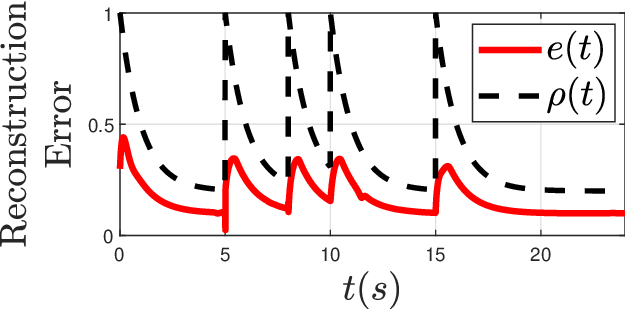}}\\
\caption{Curves of the reconstructed CBF and the reconstruction error.}
\label{fig:3}
\end{figure}
\vspace{-20pt}
\begin{align}
  & h_{\hat \phi}= \widehat {\min}\left(h_{\psi_{1}},\dots,h_{\psi_{6}}| \kappa = 10\right), \notag \\
  & h_{\psi_1} = h_{\mu_3} - 1.4 \sigma_1 + 7, \, h_{\psi_2}= h_{\mu_3} - 0.7 \sigma_2 + 7, \notag \\
  & h_{\psi_3}= h_{\psi_7} - 1.67 \sigma_3 + 10, \, h_{\psi_4}= h_{\psi_8} - 1.69 \sigma_4 + 22, \notag \\
  & h_{\psi_5}= -h_{\mu_4}, \, h_{\psi_6} = -h_{\mu_5} - \sigma_6 + 24, \notag \\
  & h_{\psi_7}= h_{\mu_2} - 2\sigma_7 + 2, \, h_{\psi_8} = \widehat {\min}\left(h_{\psi_{9}},h_{\psi_{10}}| \kappa = 10\right), \notag \\
  & h_{\psi_9} = h_{\mu_3} , \, h_{\psi_{10}} = h_{\mu_1} - 1.5\sigma_{10} + 3.
\end{align}

Besides, we set: $[\overline{t}_1, \underline{t}_1,\overline{t}_2, \underline{t}_2,\dots,\overline{t}_{10}, \underline{t}_{10}]=[0,5,0,10,0,6,0,13,0,0,0,24,6,7,13,13,13,13,13,15]$ and $\Pi _{\hat \phi}(\psi_k) = k $ for $k=1,\dots,10$. The parameters of the proposed controller are set as follows: $\lambda=10$, $c=0.01$, $\gamma = 0.01$, $\varsigma = 1$, $\rho_0 = 1$, $\rho_{\infty} = 0.2$, $\varrho = 1$, $\alpha(y)=y/2$, $Q = \mathrm{diag}\{1, l^2\}$, $\hat x(0) = z(0) = \boldsymbol{p}(0) = [1,-2]^T$, $\eta(0)=0.3$ and $\eta(T_{\hat \phi}(\mu_k))=0.1$ for any $\mu_k \in \mathcal{N}_{\hat \phi}$.

Fig.~\ref{fig:2} illustrates the trajectory of the mobile robot over the time interval $[0, 24]$. It is observed that the robot successfully visits the target regions ($\mu_1, \mu_2, \mu_3, \mu_5$) within the specified time intervals while strictly avoiding the obstacle region ($\mu_4$) at all times. This demonstrates that the robot's trajectory satisfies the given STL specification $\phi$ \eqref{eq:STL_complex}. Fig.~\ref{fig:3} presents the curves of the reconstructed CBF and the reconstruction error. It is evident that the reconstructed CBF remains non-negative at all times while the reconstruction error consistently stays within the prescribed performance bounds, indicating that the system strictly remains within the safe set.

\section{CONCLUSION} \label{sec:5}
In this work, a novel disturbance rejection control scheme is established for nonlinear systems under nested STL specifications. Utilizing the modified sTLT to guide the design process, we propose a novel recursive CBF design procedure that introduces sliding window variables to capture complex temporal relationships, thereby yielding explicit parameterized CBFs for nested STL formulas. Furthermore, a disturbance rejection QP controller based on reconstructed CBFs is proposed, ensuring that the system satisfies the CBF constraints under disturbances to ultimately accomplish the STL tasks. The reconstructed CBF-based approach requires no prior knowledge of the disturbances while relaxing the initial safety assumptions. Finally, the simulation demonstrates the effectiveness of the proposed control framework.

\addtolength{\textheight}{-12cm}   









\bibliographystyle{IEEEtran}
\bibliography{ref/ref}

\end{document}